\documentclass{llncs}
\usepackage{multicol}

\usepackage{amsmath}
\usepackage{amsfonts}
\usepackage{amssymb}
\usepackage{verbatim}

\usepackage{graphicx}
\usepackage{epic}
\usepackage{eepic}
\usepackage{epsfig,float}
\usepackage{bm}

\usepackage{multicol}

\renewcommand{\ge}{\geqslant}

\newcommand{\ol}{\overline}
\newcommand{\eps}{\varepsilon}
\newcommand{\emp}{\emptyset}

\newcommand{\Sig}{\Sigma}

\newcommand{\noin}{\noindent}

\newcommand{\bi}{\begin{itemize}}
\newcommand{\ei}{\end{itemize}}
\newcommand{\be}{\begin{enumerate}}
\newcommand{\ee}{\end{enumerate}}
\newcommand{\bd}{\begin{description}}
\newcommand{\ed}{\end{description}}

\newcommand{\bA}{{\mathbf A}}

\newcommand{\bK}{{\mathbf K}}

\newcommand{\bmA}{\bm{A}}
\newcommand{\bmB}{\bm{B}}
\newcommand{\bmF}{\bm{F}}

\newcommand{\bmK}{\bm{K}}

\newcommand{\cA}{{\mathcal A}}
\newcommand{\cB}{{\mathcal B}}

\newcommand{\cF}{{\mathcal F}}
\newcommand{\cI}{{\mathcal I}}
\newcommand{\cK}{{\mathcal K}}

\newcommand{\cS}{{\mathcal S}}

\newcommand{\fA}{{\mathfrak A}}
\newcommand{\fB}{{\mathfrak B}}

\newcommand{\fD}{{\mathfrak D}}

\newcommand{\fN}{{\mathfrak N}}

\newcommand{\rev}{\mathbb{R}}
\newcommand{\deter}{\mathbb{D}}

\newcommand{\trim}{\mathbb{T}}

\newcommand{\qedb}{\hfill$\blacksquare$}

\title{Minimal Nondeterministic Finite Automata and Atoms of Regular Languages
\thanks{This work was supported 
by the Natural Sciences and Engineering Research Council of Canada under grant No.~OGP0000871, 
the ERDF funded Estonian Center of Excellence in Computer Science, EXCS, 
and the Estonian Ministry of Education and Research target-financed 
research theme no. 0140007s12.}}
\author{Janusz~Brzozowski\inst{1} \and Hellis~Tamm\inst{2}}

\authorrunning{Brzozowski, Tamm}   

\institute{David R. Cheriton School of Computer Science, University of Waterloo, \\
Waterloo, ON, Canada N2L 3G1\\
\{{\tt brzozo@uwaterloo.ca}\}
\and
Institute of Cybernetics, Tallinn University of Technology,\\
Akadeemia tee 21, 12618 Tallinn, Estonia\\
\{{\tt hellis@cs.ioc.ee}\} 
}
\begin{document}
\maketitle

\begin{abstract}
We examine the NFA minimization problem in terms of atomic NFA's, that is, NFA's in which the right language  of every state is a union of atoms, where the atoms of a regular language  are non-empty intersections of complemented and uncomplemented left quotients of the language. We characterize all reduced atomic NFA's of a given language,  that is, those NFA's that have no equivalent states. Using atomic NFA's, we formalize Sengoku's approach to NFA minimization  and prove that his method fails to find all minimal NFA's. We also formulate the Kameda-Weiner NFA minimization in terms of quotients and atoms. 
\medskip

\noin
\end{abstract}

\section{Introduction}

Nondeterministic finite automata (NFA's) have played a major role in the theory of finite automata and regular expressions  and their applications ever since their introduction in 1959 by Rabin and Scott~\cite{RaSc59}.
In particular, the intriguing problem of finding NFA's with the minimal number of states has received much attention.
The problem was first stated by Ott and Feinstein~\cite{OtFe61} in 1961.
Various approaches have then been used over the years in attempts to answer this question; we mention a few examples here.
In 1970, Kameda and Weiner~\cite{KaWe70} studied this problem using matrices related to the states of the minimal deterministic finite automata (DFA's) for a given language and its reverse. 
In 1992, Arnold, Dicky, and Nivat~\cite{ADN92} used a ``canonical'' NFA. 
In the same year, Sengoku~\cite{Sen92} used ``normal'' NFA's and ``standard formed" NFA's.
In 1995, Matz and Potthoff~\cite{MaPo95} returned to the ``canonical'' automaton and introduced the ``fundamental'' automaton.
In 2003, Ilie and Yu~\cite{IlYu03} applied equivalence relations.
In 2005, Pol\'ak~\cite{Pol05} used the ``universal'' automaton.

Our approach is to use the recently introduced atoms and atomic 
languages~\cite{BrTa11} for this question;  
we briefly state some of their basic properties here.

The \emph{(left) quotient} of a regular language $L$ over an 
alphabet $\Sig$ by a word $w\in\Sig^*$ is the language 
$w^{-1}L=\{x\in\Sig^*\mid wx\in L\}$.
It is well known that the number of states in the complete minimal 
deterministic finite automaton recognizing $L$ is precisely 
the number of distinct quotients of $L$. 
Also, $L$ is its own quotient by the empty word $\eps$, that is $\eps^{-1}L=L$.
A \emph{quotient DFA} is a DFA uniquely determined by a regular language; its states correspond to left quotients. The quotient DFA is isomorphic to the minimal DFA.

An \emph{atom}\footnote{The definition in \cite{BrTa11} does not consider 
the intersection of all the complemented quotients to be an atom. 
Our new definition in~\cite{BrTa12} adds symmetry to the theory.}
of a regular language $L$ with quotients $K_0,\ldots, K_{n-1}$ is any 
non-empty language of the form 
$\widetilde{K_0}\cap \cdots \cap \widetilde{K_{n-1}}$, 
where $\widetilde{K_i}$ is either $K_i$ or $\ol{K_i}$, and $\ol{K_i}$ is the complement of $K_i$ with respect to $\Sig^*$.   
If the intersection with all quotients complemented is non-empty, then it constitutes the \emph{negative} atom;  
all the other atoms are \emph{positive}.
Let the number of atoms be $m$, and let the number of positive atoms be $p$. 
Thus, if the negative atom is present, $p=m-1$; otherwise, $p=m$.

So atoms of $L$ are regular languages uniquely determined by $L$. 
They are pairwise disjoint and define a partition of $\Sig^*$. 
Every quotient of $L$ (including $L$ itself)  
is a union of atoms, and  every quotient of an atom is a union of atoms.
Thus the atoms of a regular language are its basic building blocks. 
Also, $\ol{L}$ defines the same atoms as   $L$. 
The \emph{\'atomaton} is an NFA uniquely determined by a regular language; its states correspond to atoms. 
An NFA is \emph{atomic} if the right language of every state is a union of atoms.
\smallskip

Our contributions are as follows: 
\be
\item 
We characterize all trim reduced atomic NFA's of a given language, where 
an NFA is reduced if it has no equivalent states. 
\item
We show that, if $n_0$ is the minimal number of states of any NFA of a language,  then the language may have trim reduced atomic NFA's 
with as few as $n_0$ states, and as many as $2^p-1$ states.
\item
We demonstrate that the number of atomic minimal NFA's  can be 
as low as 1, or very high. For example, the language $\Sig^*ab\Sig^*$ with 3 
quotients has $281$ atomic minimal NFA's, and additional non-atomic ones.
\item
We formalize the work of Sengoku~\cite{Sen92} in our framework. He had no concept of atoms, but used an NFA equivalent to the \'atomaton 
and NFA's equivalent to atomic NFA's.
Our use of atoms  significantly clarifies Sengoku's method.
\item
We prove that Sengoku's claim that an NFA can be made atomic by adding transitions and without changing the number of states is false.
We show that there exist languages for which the minimal NFA's are all non-atomic. 
So  Sengoku's claim that his method can always find a minimal NFA is also incorrect.
\item
We formulate the Kameda-Weiner NFA minimization method~\cite{KaWe70} in terms of 
quotients and atoms.
\ee 

In Section~\ref{sec:aut} we recall some properties of automata and \'atomata.
Atomic NFA's are then presented in Section~\ref{sec:atomic}.
Sengoku's method is studied in Section~\ref{sec:Sengoku}, and
the Kameda-Weiner method, in Section~\ref{sec:KW}.
Section~\ref{sec:conc} concludes the paper.

\section{Automata and \'Atomata of Regular Languages}
\label{sec:aut}

A~\emph{nondeterministic finite automaton (NFA)} is a quintuple 
$\fN=(Q, \Sig, \eta, I,F)$, where 
$Q$ is a finite, non-empty set of \emph{states}, 
$\Sig$ is a finite non-empty \emph{alphabet}, 
$\eta:Q\times \Sig\to 2^Q$ is the  \emph{transition function},
$I\subseteq  Q$ is the set of  \emph{initial states},
and $F\subseteq Q$ is the set of \emph{final states}.
As usual, we extend the transition function to functions 
$\eta':Q\times \Sig^*\to 2^Q$, and 
$\eta'':2^Q\times \Sig^*\to 2^Q$, but
use 
$\eta$ for all three.

The \emph{language accepted} by an NFA $\fN$ is 
$L(\fN)=\{w\in\Sig^*\mid \eta(I,w)\cap F\neq \emp\}$.
Two NFA's are \emph{equivalent} if they accept the same language. 
The \emph{right language} of a state $q$  is
$L_{q,F}(\fN)=\{w\in\Sig^* \mid \eta(q,w)\cap F\neq\emp\}$.
The \emph{right language} of a set $S$ of states of $\fN$ is
$L_{S,F}(\fN)=\bigcup_{q\in S} L_{q,F}(\fN)$; so
$L(\fN)=L_{I,F}(\fN)$.
A~state is \emph{empty} if its right language is empty.
Two states   are \emph{equivalent} if their right 
languages are equal. 
An NFA is \emph{reduced} if it has no equivalent states.
The \emph{left language} of a state $q$  is
$L_{I,q}=\{w\in\Sig^* \mid q\in \eta(I,w)\}$.
A state is \emph{unreachable} if its left language is empty.
An NFA is \emph{trim} if it has no empty or unreachable states.
An NFA is \emph{minimal} if it has the minimal number of states among all
the equivalent NFA's.

A \emph{deterministic finite automaton (DFA)} is a quintuple 
$\fD=(Q, \Sig, \delta, q_0,F)$, where
$Q$, $\Sig$, and $F$ are as in an NFA, 
$\delta:Q\times \Sig\to Q$ is the transition function, 
and $q_0$ is the initial state. 

We use the following operations on automata: \\
\hglue10pt 1.
The \emph{determinization} operation $\deter$ 
applied to an NFA $\fN$ yields a DFA $\fN^{\deter}$ obtained by 
the subset construction, where only subsets reachable 
from the initial subset of $\fN^\deter$ are used, and the empty subset, 
if present, is included. \\
\hglue10pt 2.
The \emph{reversal} operation $\rev$ applied to NFA $\fN$ yields 
an NFA $\fN^{\rev}$, where the sets of initial and final states are 
interchanged and all transitions are reversed. \\
\hglue10pt 3.
The \emph{trimming} operation $\trim$ applied to an NFA deletes all unreachable and empty states.
\smallskip

The following theorem is from~\cite{Brz63}, and was also discussed in~\cite{BrTa11}: 

\begin{theorem}[Determinization]
\label{thm:Brz}
If $\fD$ is a DFA accepting a language $L$, then $\fD^{\rev\deter}$ is 
a minimal DFA for $L^R$.
\end{theorem}

Let $L$ be any non-empty regular language, and let
its set of quotients be $\cK=\{K_0,\ldots, K_{n-1}\}$. 
One of the quotients of $L$ is $L$ itself;
this is called the \emph{initial} quotient and is denoted by $K_{in}$.
A quotient is  \emph{final} if it contains the empty word $\eps$.
The set of final quotients is  $\cF=\{K_i \mid \eps \in K_i\}$.

In the following definition we use a 1-1 correspondence 
$K_i \leftrightarrow  \bmK_i$ between quotients $K_i$ of 
a language $L$ and the states ${\mathbf K}_i$ of the \emph{quotient DFA} 
$\fD$ defined below.
We refer to the ${\mathbf K}_i$ as \emph{quotient symbols}.
\begin{definition}
\label{def:quotientDFA}
\vskip-0.1cm
The \emph{quotient DFA} of $L$ is 
$\fD=(\bmK,\Sig,\delta, {\mathbf K_{in}}, \bmF)$, where
${\bmK}=\{{\mathbf K_0},\ldots,{\mathbf K_{n-1}}\}$,
${\mathbf K_{in}}$ corresponds to $K_{in}$,
$\bmF=\{{\mathbf K_i}\mid K_i \in \cF\}$, and 
$\delta({\mathbf K}_i, a)={\mathbf K}_j$ if and only if 
$a^{-1}K_i = K_j$, for all ${\mathbf K_i},{\mathbf K_j}\in \bmK$ and $a\in \Sig$.
\end{definition}

In a quotient DFA
the right language of ${\mathbf K_i}$ is $K_i$, and 
its left language 
is $\{w\in\Sig^*\mid w^{-1}L=K_i\}$.
The  language $L(\fD)$  is the right language of ${\mathbf K_{in}}$, and hence 
$L(\fD)=L$.
DFA $\fD$ is minimal, since all quotients in $K$ are distinct.

It follows from the definition of an atom, that a regular language  $L$ has at most $2^n$ atoms. 
An atom is \emph{initial} if it has $L$ (rather than $\ol{L}$) as a term;
it is \emph{final} if it contains~$\eps$.
Since $L$ is non-empty, it has at least one quotient containing~$\eps$. 
Hence it has exactly one final atom, the atom 
$\widehat{K_0}\cap \cdots \cap \widehat{K_{n-1}}$, where 
$\widehat{K_i}=K_i$ if $\eps\in K_i$, and $\widehat{K_i}=\ol{K_i}$ otherwise.
Let  $\cA=\{A_0,\ldots, A_{m-1}\}$ be the set of atoms of $L$.
By convention, $\cI$ is the set of initial atoms,  $A_{p-1}$ is the final atom and the negative atom, if present, is $A_{m-1}$.
The negative atom  is not reachable from $\cI$ and can never be final, since there must be at least one final quotient in its intersection.

As above, we use a 1-1 correspondence 
$A_i \leftrightarrow  {\mathbf A}_i$ between atoms $A_i$ of a language $L$ and 
the states ${\mathbf A}_i$ of the NFA $\fA$ defined below.
We refer to the ${\mathbf A}_i$ as \emph{atom symbols.}

\begin{definition}
\label{def:atomaton}
The \emph{\'atomaton} of $L$
 is the NFA $\fA=(\bmA,\Sig,\alpha,\bmA_I, \{\bA_{p-1}\}),$
 where $\bmA=\{{\mathbf A}_i\mid A_i\in \cA\}$,
 $\bmA_I=\{{\mathbf A}_i\mid A_i\in \cI\}$, 
 $\bA_{p-1}$ corresponds to $A_{p-1}$,
 and ${\mathbf A}_j \in \alpha({\mathbf A}_i, a)$ if and only if 
$aA_j \subseteq A_i$, for all ${\mathbf A_i},{\mathbf A_j}\in \bmA$ and $a\in\Sig$.
\end{definition}

In the \'atomaton, the right language of any state ${\mathbf A_i}$ is the atom $A_i$.
\smallskip

The results from~\cite{BrTa11} and  our definition of 
atoms in~\cite{BrTa12} imply that $\fA^\rev$ is a minimal DFA 
that accepts $L^R$. It follows from Theorem~\ref{thm:Brz} that  
$\fA^\rev$ is isomorphic to $\fD^{\rev\deter}$.
The following result from~\cite{BrTa12} makes this isomorphism precise:

\begin{theorem}[Isomorphism]
\label{thm:isomorphism}
Let $\cS$ be the collection of all subsets of the set $\bmK$ of quotient symbols.
Let $\varphi: \bmA \to \cS$ be the mapping assigning to state 
${\mathbf A}_j$, corresponding to
$A_j=K_{i_0}\cap\cdots\cap K_{i_{n-r-1}}\cap\ol{K_{i_{n-r}}}
\cap\cdots\cap \ol{K_{i_{n-1}}}$ of $\fA^\rev$, the set 
$\{\bK_{i_0},\ldots, \bK_{i_{n-r-1}}\}$.
Then $\varphi$ is a DFA isomorphism between $\fA^\rev$ and 
$\fD^{\rev\deter}$. 
\end{theorem}

\begin{corollary}
\label{cor:isomorphism}
The mapping $\varphi$ is an NFA isomorphism between 
$\fA$ and $\fD^{\rev\deter\rev}$.
\end{corollary}

\section{Atomic NFA's}
\label{sec:atomic}

A new class of NFA's was defined in~\cite{BrTa11} as follows:   

\begin{definition}
\label{def:atomic}
An NFA $\fN=(Q,\Sigma,\eta,I,F)$ is \emph{atomic} if for every  
$q\in Q$, the right language $L_{q,F}(\fN)$ of $q$ is a union of some positive  atoms of $L(\fN)$. 
\end{definition}

The following theorem, slightly restated, was proved in~\cite{BrTa11}:
\newpage

\begin{theorem}[Atomicity]
\label{thm:atomic}
A trim NFA $\fN$ is atomic if and only if $\fN^{\rev\deter}$
is minimal.
\end{theorem}

This theorem allows us to test whether an NFA $\fN$ accepting a language $L$ is atomic. To do this, reverse $\fN$ and apply the subset construction. Then $\fN$ is atomic if and only if $\fN^{\rev\deter}$
is isomorphic to the minimal DFA of $L^R$.

All three possibilities for the atomic nature of $\fN$ and $\fN^\rev$ exist:
NFA $\fN_{a}$ of Table~\ref{tab:Na} and its reverse are not atomic.
NFA $\fN_{b}$ of Table~\ref{tab:Nb} is atomic, but its reverse is not.
NFA $\fN_{c}$ of Table~\ref{tab:Nc} and its reverse are both atomic.
Note that all three of these NFA's are equivalent, and they accept $\Sig^*ab\Sig^*$.
\begin{table}[t]
\begin{minipage}[b]{0.3\linewidth}
\caption{$\fN_a$.}
\label{tab:Na}
\begin{center}
$
\begin{array}{|c| c||c| c| }    
\hline
& \ \  \ \ 
&\ \ a \ \ &\ \ b \ \   \\
\hline  
\rightarrow & 0
& \ \{0,1\} \ & \ \{0\} \   \\
\hline  
& 1
&  & \{2\}    \\
\hline  
\leftarrow& 2
&  \{2\} &  \{2\}   \\
\hline  
\end{array}
$
\end{center}
\end{minipage}
\hspace{0.3cm}
\begin{minipage}[b]{0.3\linewidth}
\caption{$\fN_b$.}
\label{tab:Nb}
\begin{center}
$
\begin{array}{|c| c||c| c| }    
\hline
& \ \  \ \ 
&\ \ a \ \ &\ \ b \ \ \\
\hline  
\rightarrow & 0
& \ \{1\} \ & \ \{0\} \ \\
\hline  
 & 1
&  \{1 \} & \  \{1,2 \} \ \\
\hline  
\leftarrow & \ 2 \
 & \ \{1,2\} \ & \ \{ 0\} \ \\
\hline  
\end{array}
$
\end{center}
\end{minipage}
\hspace{0.3cm}
\begin{minipage}[b]{0.3\linewidth}
\caption{$\fN_c$.}
\label{tab:Nc}
\begin{center}
$
\begin{array}{|c| c||c| c| }    
\hline
& \ \  \ \ 
&\ \ a \ \ &\ \ b \ \ \\
\hline  
\rightarrow & 0
& \ \{1\} \ & \ \{0\} \ \\
\hline  
 & 1
&  \{1\} & \{1,2\} \\
\hline  
\leftarrow & \ 2 \
 & \ \{2\} \ & \  \ \\
\hline  
\end{array}
$
\end{center}
\end{minipage}

\end{table}

If we allow equivalent states, there is an infinite number of atomic NFA's, 
but their behaviours are not distinct; hence we consider only reduced NFA's.
Suppose $\fB=(\cB,\Sig,\beta,\cB_I,\cB_F)$ is any trim reduced atomic NFA accepting $L$.
Since $\fB$ is atomic, the right language of any state in $\fB$ is a union of positive atoms
of $L$; hence the states of $\fB$ can be represented by sets of positive atom symbols.
Because $\fB$ is trim, it does not have a state with the empty set of atom symbols.
Since $\fB$ is reduced, no set of atom symbols appears twice.
Thus the state set $\cB$ is a collection of non-empty sets of positive atom symbols.

\begin{theorem}[Legality]
\label{thm:unions}
Suppose $L$ is a regular language,  its \'atomaton is
$\fA=(\bmA,\Sig,\alpha,\bmA_I, \{\bA_{p-1}\})$, and
$\fB=(\cB,\Sig,\beta,\cB_I,\cB_F)$ is a trim NFA, where 
$\cB=\{\bmB_1,\ldots,\bmB_r\}$ is a collection of sets of positive atom symbols and 
$\cB_I,\cB_F\subseteq\cB$.
If $\cB_i\subseteq  \cB$, define 
$S(\cB_i)=\bigcup_{\bmB_i\in \cB_i} \bmB_i$ to be the set of atom symbols 
appearing in the sets $\bmB_i$  of $\cB_i$. 
Then $\fB$ is a reduced atomic NFA of $L$ if and only if it satisfies the following
conditions:
\be
\item
\label{cond:in}
$S(\cB_I)=\bmA_I$.
\item
\label{cond:trans}
For all $\bmB_i\in \cB$, $S(\beta(\bmB_i,a))=\alpha(\bmB_i,a)$.
\item
\label{cond:out}
For all $\bmB_i\in \cB$, we have $\bmB_i\in \cB_F$ if and only if 
$\bA_{p-1}\in \bmB_i$.
\ee
\end{theorem}

Before proving the theorem, we require the following lemma:

\begin{lemma}
\label{lem:beta}
If $\fB$ satisfies Condition~\ref{cond:trans} of Theorem~\ref{thm:unions}, then
$S(\beta(\bmB_i,w)) = \alpha(\bmB_i,w)$ 
for every $\bmB_i\in \cB$ and $w\in \Sig^*$.
\end{lemma}
\begin{proof}
For $w=\eps$, we have $S(\beta(\bmB_i,\eps))=S(\bmB_i)=\bmB_i$, and 
$\alpha(\bmB_i,\eps)=\bmB_i$; so the claim holds for this case.

Assume that $S(\beta(\bmB_i,w))=\alpha(\bmB_i,w)$
for all $\bmB_i\in \cB$ and all $w\in \Sig^*$ with length less than or 
equal to $l\ge 0$. 
We prove that $S(\beta(\bmB_i,wa))=\alpha(\bmB_i,wa)$ for every $a\in \Sig$.
Let $\beta(\bmB_i,w)=\{\bmB_{i_1},\ldots,\bmB_{i_h}\}$ for
some $\bmB_{i_1},\ldots,\bmB_{i_h}\in \cB$. 
Since $\beta(\bmB_i,wa)=\beta(\beta(\bmB_i,w),a)=
\beta(\bmB_{i_1},a)\cup\cdots\cup\beta(\bmB_{i_h},a)$, we have
$S(\beta(\bmB_i,wa))=S(\beta(\bmB_{i_1},a)\cup\cdots\cup\beta(\bmB_{i_h},a))=
S(\beta(\bmB_{i_1},a))\cup\cdots\cup S(\beta(\bmB_{i_h},a))$.
By   Condition~\ref{cond:trans}, the latter is equal 
to $\alpha(\bmB_{i_1},a)\cup\cdots\cup \alpha(\bmB_{i_h},a)=
\alpha(\bmB_{i_1}\cup\cdots\cup\bmB_{i_h},a)=
\alpha(S(\beta(\bmB_i,w)),a)$.
By the inductive assumption, we get
$\alpha(S(\beta(\bmB_i,w)),a)=\alpha(\alpha(\bmB_i,w),a)=
\alpha(\bmB_i,wa)$, which proves our claim.
\qed
\end{proof}
\noin
{\bf Proof of Theorem~\ref{thm:unions}}
\begin{proof}
First we prove that any NFA $\fB$ satisfying
Conditions~\ref{cond:in}--\ref{cond:out} is an atomic NFA of $L$.
Let $\bmB_i\in\cB$ be a state of $\fB$. 
If $w\in L_{ \bmB_i, \cB_F} (\fB)$, then by 
Condition~\ref{cond:out}, there exists 
$\bmB_j\in\beta(\bmB_i,w)$ such that $\bA_{p-1}\in \bmB_j$, and 
we have $\bA_{p-1}\in S(\beta(\bmB_i,w))$. 
By Lemma~\ref{lem:beta}, we get $\bA_{p-1}\in \alpha(\bmB_i,w)$, 
implying that there is some $\bA_k\in\bmB_i$ such that 
$w\in L_{\bA_k,\{\bA_{p-1}\}}(\fA)$. 
Conversely, if $w\in L_{\bA_k,\{\bA_{p-1}\}}(\fA)$ and $\bA_k\in\bmB_i$, then 
$\bA_{p-1}\in\alpha(\bmB_i,w)=S(\beta(\bmB_i,w))$. 
Hence there exists $\bmB_j\in\beta(\bmB_i,w)$ such that $\bA_{p-1}\in \bmB_j$. 
Consequently, every word accepted in $\fB$ from state $\bmB_i$ is 
in some atom $A_k$ such that $\bA_k\in \bmB_i$, and 
every word in an atom $A_k$ such that $\bA_k\in \bmB_i$, is also in 
$L_{ \bmB_i, \cB_F} (\fB)$.
Therefore the right language of $\bmB_i$ in $\fB$ is equal to 
the union of atoms $A_k$ such that $\bA_k\in \bmB_i$.
In particular, $L_{\cB_I,\cB_F} (\fB)$ is the union of atoms whose atom symbols
appear in the initial collection of $\fB$ which, by Condition~\ref{cond:in},  
is the same as the union of atoms whose atom symbols are initial in $\fA$.
But that last union is precisely $L_{\bmA_I,\{\bA_{p-1}\}}(\fA)=L$.
Since any two sets $\bmB_i$ and $\bmB_j$ are different, and  
atoms are disjoint, $\fB$ is reduced.
Hence $\fB$ is a reduced atomic NFA of $L$.

Conversely, we show that if $\fB$ is a reduced atomic NFA of $L$, 
then it must satisfy Conditions~\ref{cond:in}--\ref{cond:out}.
So in the following we assume that $\fB$ is atomic, that is, 
for every state $\bmB_i$ of $\fB$, the right language of $\bmB_i$ 
is equal to the union of atoms $A_k$ such that $\bA_k\in \bmB_i$.

First, we show that Condition~\ref{cond:in} holds.
Let $\bA_j\in S(\cB_I)$. Then there is a state $\bmB_j\in\cB_I$ such that 
$\bA_j\in \bmB_j$. So for any $w\in A_j$, $w\in L(\fB)$. 
Since $L(\fB)=L(\fA)$, we have $w\in L(\fA)$ for all $w\in A_j$.
Thus $\bA_j\in \bmA_I$.
Conversely, if $\bA_j\in \bmA_I$, then for all $w\in A_j$, $w\in L(\fA)=L(\fB)$. 
Since $\fB$ is atomic, there is an initial state $\bmB_j$ such that 
$A_j\subseteq L_{\bmB_j,\cB_F}(\fB)$. Hence $\bA_j\in S(\cB_I)$. 

Next, we prove Condition~\ref{cond:trans}. 
If $\bA_j\in S(\beta(\bmB_i,a))$, then $L_{\bmB_i,\cB_F}(\fB)$ must contain $aA_j$. 
So there must exist some $\bA_i\in \bmB_i$ such that $aA_j \subseteq A_i$.
Thus $\bA_j\in \alpha(\bmB_i,a)$.
Conversely, if $\bA_j\in \alpha(\bmB_i,a)$, then there is an atom 
$\bA_i\in \bmB_i$ such that $\bA_j \in \alpha(\bA_i,a)$, implying $aA_j \subseteq A_i$.
Since $\bA_i\in \bmB_i$, $L_{\bmB_i,\cB_F}(\fB)$ must contain $aA_j$.
Hence $\bA_j\in S(\beta(\bmB_i,a))$.

To show that Condition~\ref{cond:out} holds, we first suppose that
$\bmB_i\in \cB_F$. Then $\eps$ is in the right language of $\bmB_i$.
Since $\fB$ is atomic, $\eps$ must be in one of the atoms of $\bmB_i$.
However, the only atom containing $\eps$ is $A_{p-1}$, so $\bA_{p-1}\in\bmB_i$.
Conversely, if $\bA_{p-1}\in \bmB_i$, then $\eps$ is in the right language of  
$\bmB_i$, and $\bmB_i$ is a final state by definition of an NFA.
\qed
\end{proof}

\begin{example}
\label{ex:atomconstr}
Consider the trim \'atomaton $\fA^\trim$ of Table~\ref{tab:deflegal1} and 
the atomic NFA $\fB$ of Table~\ref{tab:deflegal2}.
Here $\cB=\{\bmB_0,\bmB_1,\bmB_2 \}$, where $\bmB_0=\{\bA_0,\bA_1\}$,
$\bmB_1=\{\bA_2\}$, and 
$\bmB_2=\{\bA_0,\bA_2\}$.
The initial collection is $\cB_I=\{\bmB_0 \}=\{ \{\bA_0,\bA_1\} \}$, and the final
collection is 
$\cB_F=\{\bmB_1,\bmB_2 \}=\{ \{\bA_2\},\{\bA_0,\bA_2\} \}$.
One verifies that all the conditions of Theorem~\ref{thm:unions} hold,
and  NFA's $\fA^\trim$ and $\fB$ are equivalent.
\begin{table}[t]
\begin{minipage}[b]{0.37\linewidth}
\caption{\'Atomaton $\fA^\trim$.}
\label{tab:deflegal1}
$
\begin{array}{|c| c||c| c| c|c|}    
\hline
& \ \  \ \ 
&\ \ a \ \ &\ \ b \ \   \\
\hline  
\rightarrow & \ \bA_0 \
&  \ \{\bA_0,\bA_1\} \  & \ \{\bA_0,\bA_2\} \    \\
\hline  
\rightarrow& \bA_1
&  \  \{\bA_2\} \  &   \\
\hline  
\leftarrow& \bA_2
&  \   \  &   \\
\hline
\end{array}
$
\end{minipage}
\hspace{.5cm}
\begin{minipage}[b]{0.55\linewidth}
\caption{Atomic NFA $\fB$.}
\label{tab:deflegal2}
$
\begin{array}{|c| c||c| c| }    
\hline
& \ \  \ \ 
&\ \ a \ \ &\ \ b \ \ \\
\hline  
\rightarrow & \{\bA_0,\bA_1\}
& \ \{ \{\bA_0,\bA_1\},\{\bA_2\} \} \ & \ \{ \{\bA_0,\bA_2\} \}\ \\
\hline  
\leftarrow& \ \{\bA_2\} \
&  & \\
\hline  
\leftarrow & \{\bA_0,\bA_2\}
& \ \{ \{\bA_0,\bA_1\} \} \ & \ \{ \{\bA_0,\bA_2\}\} \ \\
\hline  
\end{array}
$
\end{minipage}
\end{table}
\qedb
\end{example}

The number of trim reduced atomic NFA's can be very large. 
There can be such NFA's with as many as $2^p-1$ non-empty states, 
since there are that many non-empty sets of positive atoms. 
However, in a general case, not all sets of positive atom symbols can be states of 
an atomic NFA. 
The largest reduced atomic NFA is characterized in the following theorem.

\begin{theorem}[Maximal atomic NFA]
\label{thm:subsets}
If $\cB$ is the collection of all sets $\bmB_i$ such that 
$\bmB_i$ is a non-empty subset of the set of positive atom symbols 
$\{\bA_h \mid A_h\subseteq K_j\}$ of any quotient $K_j$ of $L$, 
then there exists a trim reduced atomic NFA of $L$ with  state set $\cB$.
\end{theorem}
\begin{proof}
Let $\fB=(\cB,\Sig,\beta,\cB_I,\cB_F)$ be an NFA in which the state set 
$\cB$ is the collection of all sets $\bmB_i$ such that 
$\bmB_i$ is a non-empty subset of the set of atom symbols 
$\{\bA_h \mid A_h\subseteq K_j\}$ of any quotient $K_j$ of $L$, 
where $j\in\{0,\ldots,n-1\}$,
$\beta(\bmB_i,a)= \{\bmB_j \mid \bmB_j\subseteq\alpha(\bmB_i,a)\}$
for every $\bmB_i\in\cB$ and $a\in\Sig$,
$\bmB_i\in\cB_I$ if and only if $\bmB_i$ is a subset of the set of atom symbols 
of the initial quotient $K_{in}$, and 
$\bmB_i\in\cB_F$ if and only if $\bA_{p-1}\in\bmB_i$. 
We claim that $\fB$ is a trim reduced atomic NFA of $L$.

First, we show that $\fB$ is trim. 
Let us consider any state $\bmB_i$ of $\fB$.
Let $K_j$ be a quotient such that $\bmB_i$ is a subset of the set of atom symbols
of $K_j$, and let $\bmB_j$ be the set of atom symbols corresponding to $K_j$.
Let $\bmB_0$ be the set of atom symbols corresponding to the initial
quotient $K_{in}$ of $L$. Note that $\bmB_0=\bmA_I$.
Since every set of atom symbols corresponding to some quotient is reachable
from the initial set of atom symbols in the \'atomaton $\fA$, 
there must be a word $w\in\Sig^*$, such that $\bmB_j$ is reachable 
from $\bmB_0$ by $w$ in $\fA$. 
We show that $\bmB_i$ is reachable from some initial state of $\fB$ by $w$.
If $w=\eps$, then $K_j=K_{in}$, and since $\bmB_i\subseteq\bmB_j$, it follows 
that $\bmB_i$ is an initial state of $\fB$ reachable from itself 
by $\eps$.
If $w=ua$ for some $u\in\Sig^*$ and $a\in\Sig$, then there is a state
$\bmB_u$ of $\fB$, reachable from $\bmB_0$ by $u$, such that 
$\bmB_u$ corresponds to the quotient $u^{-1}L$ of $L$ and 
$\bmB_j=\alpha(\bmB_u,a)$. 
Since $\bmB_i\subseteq\bmB_j$ and $\bmB_j=\alpha(\bmB_u,a)$, 
by the definition of $\beta$ we have 
$\bmB_i\in\beta(\bmB_u,a)$. 
Thus, $\bmB_i$ is reachable from $\bmB_0$ in $\fB$ by $ua$.

We also have to show that there is a word $w\in\Sig^*$, such that some 
final state of $\fB$ is reachable from $\bmB_i$ by $w$.
If $\bmB_i$ is final, then it is reachable from itself by $w=\eps$.
If $\bmB_i$ is not final, then let us consider any $\bA_k\in\bmB_i$. 
Since the right language of the state $\bA_k$ in the \'atomaton $\fA$ 
is not empty, and $\bA_k$ cannot be the final state of $\fA$, there must be 
some state $\bA_l$ of $\fA$ and some $a\in\Sig$, such that 
$\bA_l\in\alpha(\bA_k,a)$.
Now we know that there is some $\bmB_j$ such that $\bA_l\in\bmB_j$ and
$\alpha(\bmB_i,a)=\bmB_j$. Since $\beta(\bmB_i,a)$ is the collection
of all non-empty subsets of $\bmB_j$, it follows that 
$\{\bA_l\}\in\beta(\bmB_i,a)$.
Since the final state $\bA_{p-1}$ of $\fA$ is reachable from $\bA_l$ by 
any word $v\in A_l$, we get 
$\{\bA_{p-1}\}\in\beta(\bmB_i,av)$ by the definition of $\beta$. 
So a final state $\{\bA_{p-1}\}$ of $\fB$ is reachable from $\bmB_i$ by $av$.   
Thus, $\fB$ is trim.

To see that $\fB$ is a reduced atomic NFA, one verifies that 
Conditions~\ref{cond:in}--\ref{cond:out} of Theorem~\ref{thm:unions} hold.
Thus by Theorem~\ref{thm:unions}, $\fB$ is a trim reduced atomic NFA of $L$.
\qed
\end{proof}

\begin{theorem}[NFA with $2^p-1$ states]
\label{thm:maximal}
A regular language $L$ has a trim reduced atomic NFA with $2^p-1$ states if and 
only if for some quotient $K_i$ of $L$, $K_i= A_0 \cup \cdots \cup A_{p-1}$.
\end{theorem}

\begin{proof}
Let $\fB=(\cB,\Sig,\beta,\cB_I,\cB_F)$ be a trim reduced atomic NFA of $L$ 
with $2^p-1$ states. Then there must be a state $\bmB_i$ of $\fB$ such that
$\bmB_i=\{\bA_0,\ldots,\bA_{p-1}\}$. Since the right language of any state 
of a trim NFA is a subset of some quotient, we have
$L_{\bmB_i,\cB_F}(\fB)=A_0\cup\cdots\cup A_{p-1}\subseteq K_i$ for some
quotient $K_i$ of $L$.
On the other hand, $K_i$ must be a union of some positive atoms, 
so we get $K_i= A_0 \cup \cdots \cup A_{p-1}$.

Conversely, let $K_i= A_0 \cup \cdots \cup A_{p-1}$ be a quotient of $L$
which includes all the positive atoms of $L$. Then 
by Theorem~\ref{thm:subsets}, there is a trim reduced atomic NFA of $L$
in which the state set is the collection of all non-empty subsets of the set 
of positive atom symbols. This NFA has $2^p-1$ states. 
\qed
\end{proof}

The construction of reduced atomic NFA's is illustrated in the following example.
To simplify the notation, we do not use atom symbols in examples.

\begin{example}
\label{ex:reducedatomic}
Consider the minimal DFA $\fD$ taken from~\cite{KaWe70} and  shown in 
Table~\ref{tab:dkw}.
It accepts the language $L=\Sig^*(b\cup aa) \cup a$, and its quotients are
$K_0=\eps^{-1}L=L$, 
$K_1=a^{-1}L=\Sig^*(b\cup aa) \cup a \cup \eps$, and
$K_2=b^{-1}L=\Sig^*(b\cup aa) \cup \eps$.
NFA $\fD^{\rev\deter\rev\trim}$ and the isomorphic trim \'atomaton $\fA^\trim$ with states renamed  are shown in Tables~\ref{tab:drdrkw} and~\ref{tab:akw}.
The positive atoms are
$A=\Sig^*(b\cup aa)$, $B=a$ and $C=\eps$, and
$K_0=A\cup B$, 
$K_1=A\cup B\cup C$,
and $K_2=A \cup C$.

\begin{table}[b]
\begin{minipage}[b]{0.25\linewidth}
\caption{$\fD$.}
\label{tab:dkw}
\begin{center}
$
\begin{array}{|c| c||c| c| }    
\hline
& \ \  \ \ 
&\ \ a \ \ &\ \ b \ \ \\
\hline  
\rightarrow & 0
& \ 1 \ & \ 2 \ \\
\hline  
\leftarrow& 1
&  1  & 2 \\
\hline  
\leftarrow & 2
 &  0 &  2  \\
\hline  
\end{array}
$
\end{center}
\end{minipage}
\hspace{0.05cm}
\begin{minipage}[b]{0.25\linewidth}
\caption{$\fD^{\rev\deter\rev\trim}$.}
\label{tab:drdrkw}
\begin{center}
$
\begin{array}{|c| c||c| c| }    
\hline
& \ \  \ \ 
&\ \ a \ \ &\ \ b \ \   \\
\hline  
\leftarrow & 12
& &  \\
\hline  
\rightarrow & 01
&  \{12\} &   \\
\hline  
\rightarrow & \ 012 \
&  \{012,01\}  & \ \{012,12\} \  \\
\hline  
\end{array}
$
\end{center}
\end{minipage}
\hspace{1.5cm}
\begin{minipage}[b]{0.25\linewidth}
\caption{$\fA^ \trim$.}
\label{tab:akw}
\begin{center}
$
\begin{array}{|c| c||c| c| }    
\hline
& \ \  \ \ 
&\ \ a \ \ &\ \ b \ \   \\
\hline  
\leftarrow & C
& &  \\
\hline  
\rightarrow & B
&  \{C\} &   \\
\hline  
\rightarrow & \ A \
& \ \{A,B\} \ & \ \{A,C\} \  \\
\hline  
\end{array}
$
\end{center}
\end{minipage}
\end{table}

\begin{table}[hbt]
\begin{minipage}[b]{0.45\linewidth}
\caption{NFA $\fB_1$.}
\label{tab:b1}
\begin{center}
$
\begin{array}{|c| c||c| c| c|c|}    
\hline
& \ \  \ \ 
&\ \ a \ \ &\ \ b \ \   \\
\hline  
\rightarrow & \ \{A,B\} \
&  \ \{A,B\},\{A,C\} \  & \ \{A,C\} \    \\
\hline  
\leftarrow& \{A,C\}
&  \  \{A,B\} \  &  \{A,C \} \\
\hline  
\end{array}
$
\end{center}
\end{minipage}
\hspace{1cm}
\begin{minipage}[b]{0.45\linewidth}
\caption{Atomic NFA $\fB_2$.}
\label{tab:b2}
\begin{center}
$
\begin{array}{|c| c||c| c| }    
\hline
& \ \  \ \ 
&\ \ a \ \ &\ \ b \ \ \\
\hline  
\rightarrow & \{A,B\}
& \ \{A,B\},\{C\} \ & \ \{A,C\} \ \\
\hline  
\leftarrow& \ \{C\} \
&  & \\
\hline  
\leftarrow & \{A,C\}
 &  \{A,B\} &  \{A,C\}  \\
\hline  
\end{array}
$
\end{center}
\end{minipage}
\end{table}

\begin{table}[t]
\begin{minipage}[b]{0.45\linewidth}
\caption{A 5-state NFA.}
\label{tab:b3}
\begin{center}
$
\begin{array}{|c| c||c| c| }    
\hline
& \ \  \ \ 
&\ \ a \ \ &\ \ b \ \ \\
\hline  
\rightarrow & \{A\}
& \ \{A\},\{B\} \ & \ \{A,C\} \ \\
\hline  
\rightarrow& \ \{B\} \
& \{C\} & \\
\hline  
\leftarrow & \{A,C\}
 &  \{A,B\} &  \{A,C\}  \\
\hline  
\leftarrow& \ \{C\} \
&  & \\
 \hline  
 & \{A,B\}
 & \ \{A,B\},\{C\} \ &  \{A\},\{C\}  \\
\hline  
\end{array}
$
\end{center}
\end{minipage}
\hspace{.4cm}
\begin{minipage}[b]{0.45\linewidth}
\caption{A 7-state NFA.}
\label{tab:b4}
\begin{center}
$
\begin{array}{|c| c||c| c| c|c|}    
\hline
& \ \  \ \ 
&\ \ a \ \ &\ \ b \ \   \\
\hline  
\rightarrow & \ \{A\} \
&  \ \{A\},\{B\} \  &  \{A,C\}     \\
\hline  
\rightarrow& \{B\}
&  \{C\} &  \\
\hline  
\leftarrow& \{A,C\}
&  \  \{A,B\} \  &  \{A,C\} \\
\hline  
\leftarrow& \{C\}
&    &   \\
\hline  
\rightarrow& \{A,B\}
&   \{A,B,C\}, \{B,C\}  &  \{A,C\} \\
\hline  
\leftarrow& \{A,B,C\}
&    \{A,B,C\},\{B,C\}   &  \{A,C\} \\
\hline  
\leftarrow& \{B,C\}
&  \  \{C\} \  &   \\
\hline  
\end{array}
$
\end{center}
\end{minipage}
\end{table}

Since the set $\{A,B\}$ of initial atoms does not contain all positive atoms, no 1-state NFA exists.
\be
\item 
For the initial state we could pick one state $\{A,B\}$ with two atoms.  From there, the \'atomaton reaches 
$\{A,B,C\}$ under $a$, and  $\{A,C\}$ under $b$. 
        \be
        \item
If we pick $\{A,C\}$
as the second state,  we can cover $\{A,B,C\}$ by $\{A,B\}$ and 
$\{A,C\}$, as  in Table~\ref{tab:b1}.
Here the minimal atomic NFA is unique.
        \item
        We can also use $\{A,B,C\}$ as a state. Then we need $\{A,C\}$
        for the transition under $b$. This gives an NFA  isomorphic to the DFA of Table~\ref{tab:dkw}.
        \item
        We can use state $\{C\}$
        as shown in Table~\ref{tab:b2}.
        \ee
\item
We can pick two initial states, $\{A\}$ and $\{B\}$. 
        \be
        \item
        If we add $\{C\}$, this leads to the  \'atomaton of Table~\ref{tab:akw}.
        \item
        A 5-state solution is shown in Table~\ref{tab:b3}.
        \ee
\item
We can use three initial states, $\{A\}$, $\{B\}$ and $\{A,B\}$. 
        A 7-state NFA is shown in  Table~\ref{tab:b4}. This 
           is a largest possible reduced solution.\qedb
\ee

\end{example}

The number of minimal atomic NFA's can also be very large. 
\begin{example}
\label{ex:atomicminimal}
Let $\Sig=\{a,b\}$ and consider the language $L=\Sig^*a\Sig^*b\Sig^*=\Sig^*ab\Sig^*$.
The quotients of $L$ are $K_0=L$, $K_1=L\cup b\Sig^*$ and $K_2=\Sig^*$.
The quotient DFA of $L$ is shown in Table~\ref{tab:d}, and its \'atomaton, in Tables~\ref{tab:a} and~\ref{tab:a_relabel} (where the atoms have been relabelled). 
The atoms  are $A=L$, $B=b^*ba^*$ and $C=a^*$, and there is no negative atom.
Thus the quotients are $K_0=L=A$, $K_1=A\cup B$, and $K_2=A\cup B\cup C$.

We find all the minimal atomic NFA's of $L$.
Obviously, there is no 1-state solution.
The states of any atomic NFA are sets of atoms, and 
there are seven non-empty sets of atoms to choose from. 
Since there is only one initial atom, there is no choice: we must take $\{A\}$.
For the transition $(A,a,\{A,B\})$, we can add $\{B\}$ or $\{A,B\}$. 
If there are only two states, atom $\{C\}$ cannot be reached. So there is no  2-state atomic NFA.
The results for 3-state atomic NFA's  are summarized in Proposition~\ref{prop:281}.

\begin{table}[hbt]
\begin{minipage}[b]{0.3\linewidth}
\caption{DFA $\fD$.}
\label{tab:d}
\begin{center}
$
\begin{array}{|c|c|| c|c|}    
\hline
 & & \  a \
& \ b \  \\
\hline
\hline
\rightarrow & 0 & 1
& 0  \\
\hline  
 & 1 & 
1 &    2  \\
\hline  
\leftarrow & \ 2 \ & \ 2 \
& \ 2 \ \\
\hline  
\end{array}
$
\end{center}
\end{minipage}
\hspace{0.1cm}
\begin{minipage}[b]{0.3\linewidth}
\caption{\'Atomaton $\fA$.}
\label{tab:a}
\begin{center}
$
\begin{array}{|c| c||c| c| }    
\hline
& \ \  \ \ 
&\ \ a \ \ &\ \ b \ \   \\
\hline  
\leftarrow & 2
& \{2 \} &  \\
\hline  
 & 12
&   & \{12,2\}  \\
\hline  
\rightarrow & \ 012 \
&  \{012,12\}  & \ \{012\} \  \\
\hline  
\end{array}
$
\end{center}
\end{minipage}
\hspace{0.3cm}
\begin{minipage}[b]{0.3\linewidth}
\caption{$\fA$ relabelled.}
\label{tab:a_relabel}
\begin{center}
$
\begin{array}{|c| c||c| c| }    
\hline
& \ \  \ \ 
&\ \ a \ \ &\ \ b \ \   \\
\hline  
\leftarrow & C
& \{C\} &  \\
\hline  
 & B
&   &  \ \{B,C\} \ \\
\hline  
\rightarrow & \ A \
& \ \{A,B\} \ & \ \{A\} \  \\
\hline  
\end{array}
$
\end{center}
\end{minipage}
\end{table}

\begin{table}[hbt]
\begin{minipage}[b]{0.45\linewidth}
\caption{NFA $\fN_2$.}
\label{tab:fn1}
\begin{center}
$
\begin{array}{|c| c||c| c| }    
\hline
& \ \  \ \ 
&\ \ a \ \ &\ \ b \ \ \\
\hline  
\rightarrow & A
& \ AB \ & \ A \ \\
\hline  
 & AB
&  AB  & AB,C \\
\hline  
\leftarrow & \ C \
 & \ C \ & \  \ \\
\hline  
\end{array}
$
\end{center}
\end{minipage}
\hspace{0.2cm}
\begin{minipage}[b]{0.45\linewidth}
\caption{NFA $\fN_9$.}
\label{tab:fn9}
\begin{center}
$
\begin{array}{|c| c||c| c| c|c|}    
\hline
& \ \  \ \ 
&\ \ a \ \ &\ \ b \ \   \\
\hline  
\rightarrow & \ A \
&  \ A,AB \  &  \  A  \   \\
\hline  
 & AB
&  \  A,AB \  & \ A,AB,C \ \\
\hline  
\leftarrow & C
&   \ C  \  &      \\
\hline  
\end{array}
$
\end{center}
\end{minipage}
\end{table}

\begin{proposition} 
\label{prop:281}
The language $\Sig^*ab\Sig^*$ has 281 minimal atomic NFA's. 
\end{proposition}
\begin{proof}
We concentrate on 3-state solutions. 
We drop the curly brackets and commas and represent sets of atoms by words. Thus $\{A,AB
,BC\}$ stands for 
$\{ \{A\}, \{A,B\},\{B,C\} \}$.

State $A$ is the only initial state and so it must be included.
To implement the transition
$(A,a,\{A,B\})$ from $\fA$,
either $B$ or $AB$ must be chosen. 
\be
\item
If $B$ is chosen, then there must be a set containing $C$ but not $A$; otherwise 
the transition 
 $(B,b,\{B,C\})$ cannot be realized.
        \be
        \item
        If $BC$ is taken, then $C$ must be taken, and this would make four states.
        \item
         Hence $C$ must be chosen, giving states $A$, $B$, and $C$.
         This yields the \'atomaton $\fA=\fN_1$.
         \ee
\item
If $AB$ is chosen, then we could choose $C$, $AC$ or $ABC$, since $BC$ would also require
 $C$. Thus there are three cases:
        \be
        \item
        $\{A,AB,C\}$ yields $\fN_2$ of Table~\ref{tab:fn1}, if the minimal number of 
        transitions is used. 
        The following transitions can also be added: $(A,a,A)$, $(AB,a,A)$, $(AB,b,A)$.
        Since these can be added independently, we have eight more NFA's. 
        Using the maximal number of transitions, we get $\fN_9$ of Table~\ref{tab:fn9}.
        \item
        $\{A,AB,AC\}$ results in $\fN_{10}$ with the minimal number of transitions, and 
        $\fN_{25}$ with the maximal one.
        \item
        $\{A,AB,ABC\}$ results in $\fN_{26}$ (the quotient DFA) with the minimal number 
of transitions, and  $\fN_{281}$ with the maximal one.
        \ee
\ee

\begin{table}[hbt]
\begin{minipage}[b]{0.45\linewidth}
\caption{NFA $\fN_{10}$.}
\label{tab:fn10}
\begin{center}
$
\begin{array}{|c| c||c| c| }    
\hline
& \ \  \ \ 
&\ \ a \ \ &\ \ b \ \ \\
\hline  
\rightarrow & A
& \ AB \ & \ A \ \\
\hline  
 & AB
&  AB  & \  AB,AC \ \\
\hline  
\leftarrow & \ AC \
 & \ AB,AC \ & \ A \ \\
\hline  
\end{array}
$
\end{center}
\end{minipage}
\hspace{0.2cm}
\begin{minipage}[b]{0.45\linewidth}
\caption{NFA $\fN_{25}$.}
\label{tab:fn25}
\begin{center}
$
\begin{array}{|c| c||c| c| c|c|}    
\hline
& \ \  \ \ 
&\ \ a \ \ &\ \ b \ \   \\
\hline  
\rightarrow & \ A \
&  \ A,AB \  &  \  A  \   \\
\hline  
 & AB
&  \  A,AB \  & \ A,AB,AC \ \\
\hline  
\leftarrow & AC
&   \ A,AB,AC  \  &   A   \\
\hline  
\end{array}
$
\end{center}
\end{minipage}
\end{table}

\begin{table}[h]
\begin{minipage}[b]{0.45\linewidth}
\caption{NFA $\fN_{26}$.}
\label{tab:fn26}
\begin{center}
$
\begin{array}{|c| c||c| c| }    
\hline
& \ \  \ \ 
&\ \ a \ \ &\ \ b \ \ \\
\hline  
\rightarrow &A
& \ AB \ & \ A \ \\
\hline  
 & AB
&  AB  & \  ABC \ \\
\hline  
\leftarrow & \ ABC \
 & \ ABC \ & \ ABC \ \\
\hline  
\end{array}
$
\end{center}
\end{minipage}
\hspace{0.2cm}
\begin{minipage}[b]{0.45\linewidth}
\caption{NFA $\fN_{281}$.}
\label{tab:fn281}
\begin{center}
$
\begin{array}{|c| c||c| c| c|c|}    
\hline
& \ \  \ \ 
&\ \ a \ \ &\ \ b \ \   \\
\hline  
\rightarrow & \ A \
&  \ A,AB \  &  \ A  \   \\
\hline  
 & AB
&  \  A,AB \  & \ A,AB,ABC \ \\
\hline  
\leftarrow & ABC
&   \ A,AB,ABC \  &   A,AB,ABC   \\
\hline  
\end{array}
$
\end{center}
\end{minipage}
\end{table}
\begin{table}[htb]
\caption{NFA $\fN_{282}$.}
\label{tab:na282}
\begin{center}
$
\begin{array}{|c|c|| c|c|}    
\hline
 & & \  a \
& \ b \  \\
\hline
\hline
\rightarrow & 0 
& 1 & 0  \\
\hline  
 & 1 
 & 1 & \  0,1,2 \ \\
\hline  
\leftarrow & \ 2 \ 
& \ 0,2 \ &  \\
\hline  
\end{array}
$
\end{center}
\end{table}

As well, $L$ has 3-state non-atomic NFA's.
The determinized version of NFA $\fN_{10}$ of Table~\ref{tab:fn10} is not minimal.
By Theorem~\ref{thm:atomic}, $\fN_{10}^\rev$ is not atomic. But $L^R=\Sig^*ba\Sig^*$;
hence we obtain a non-atomic 3-state NFA for $L$ by reversing $\fN_{10}$ and interchanging $a$ and $b$.
That NFA with renamed states is shown in Table~\ref{tab:na282}.

The right languages of the states of $\fN_{282}$ are:
$L_0=L=A$, 
$L_1=A\cup B$, and
$L_2=\eps \cup a  \cup aa\Sig^*\cup abb^*aa^*b\Sig^*$, which is not a union of atoms.
Six more non-atomic NFA's can be derived from NFA's between $\fN_{10}$ and  $\fN_{25}$.
\qed
\end{proof}

This is a rather large number of NFA's for a language with  3 quotients. 
\qedb
\end{example}

One can verify that there is no NFA with fewer than 3 states which
accepts the language $L=\Sig^*ab\Sig^*$.
This implies that every minimal atomic NFA of $L$ is also 
a minimal NFA of $L$.
However, this is not the case with all regular languages, as we will see 
in the next section.

\section{Sengoku's NFA Minimization Method}
\label{sec:Sengoku}

Sengoku had no concept of atom, but he came very close to discovering it.
For a language accepted by a minimal DFA $\fD$, the \emph{normal} 
NFA~\cite{Sen92}(p.~18) is isomorphic to $\fD^{\rev\deter\rev\trim}$, and hence to 
the trim \'atomaton, by our Corollary~\ref{cor:isomorphism}. 
Moreover, he defines an NFA $\fN$ to be in \emph{standard form}~\cite{Sen92}(p.~19) 
if $\fN^{\rev\deter}$ is minimal.
By our Theorem~\ref{thm:atomic}, such an $\fN$ is atomic.
Sengoku makes the following claim~\cite{Sen92}(p.~20):
\begin{quote}
\vskip-0.1cm
\emph{We can transform the nondeterministic automaton into its standard form 
by adding some extra transitions to the automaton. Therefore the number of 
states is unchangeable.}
\end{quote}
\vskip-0.1cm
This claim amounts to stating that any NFA can be transformed to an equivalent  
atomic NFA by adding some transitions. Unfortunately, the claim is false:
\begin{theorem}
\label{thm:Sengoku}
There exists a language for which no minimal NFA is atomic.
\end{theorem}
\begin{proof}
\vskip-0.1cm
This example is from~\cite{MaPo95}. 
A quotient DFA $\fD$, the NFA $\fD^{\rev\deter\rev}$, and its isomorphic \'atomaton 
$\fA$ with relabelled states are  in Tables~\ref{tab:d_mp}--\ref{tab:a_mp}, 
respectively (there is no negative atom). We  now drop the curly brackets and commas in tables, and 
represent sets of atoms by words.
A minimal NFA $\fN_{min}$ of this language, having four states, is shown in 
Table~\ref{tab:n_mp}; it is not atomic and it is not unique. 
We try to construct a 4-state atomic NFA $\fN_{atom}$ equivalent to $\fD$. 

\vskip-0.5cm
\begin{table}[hbt]
\begin{minipage}[b]{0.25\linewidth}
\caption{$\fD$.}
\label{tab:d_mp}
\begin{center}
$
\begin{array}{|c|c|| c|c|}    
\hline
 & & \  a \
& \ b \  \\
\hline
\hline
\rightarrow & 0 
& 1 & 2  \\
\hline  
 & 1 
 &  3 &    4  \\
\hline  
\leftarrow & \ 2 \ 
& \ 5 \ & \ 4 \ \\
\hline  
 & 3 & 
3 &    1  \\
\hline  
 & 4 & 
6 &    2  \\
\hline  
\leftarrow & 5 & 
7 &    2  \\
\hline  
 & 6 & 
3 &    8  \\
\hline  
\leftarrow & 7 & 
7 &    7  \\
\hline  
 & 8 & 
6 &    7  \\
\hline  
\end{array}
$
\end{center}
\end{minipage}
\hspace{0.05cm}
\begin{minipage}[b]{0.35\linewidth}
\caption{$\fD^{\rev\deter\rev}$.}
\label{tab:drdr_mp}
\begin{center}
$
\begin{array}{|c| c||c| c| }    
\hline
& \ \  \ \ 
&\ \ a \ \ &\ \ b \ \   \\
\hline  
\leftarrow & 257
& 257,04578 &  \\
\hline  
\rightarrow & 04578
&  12678 &  \ 257 \ \\
\hline  
 & 12678 
& \  \ & 04578,03-8   \\
\hline  
\rightarrow & 03-8
&   &  \ 12678 \ \\
\hline  
 & 1-8
 &  03-8 &   \\
\hline  
\rightarrow & 0-8
&  1-8,0-8 &  \ 1-8,0-8 \ \\
\hline  
\end{array}
$
\end{center}
\end{minipage}
\hspace{0.8cm}
\begin{minipage}[b]{0.3\linewidth}
\caption{ $\fA$.}
\label{tab:a_mp}
\begin{center}
$
\begin{array}{|c| c||c| c| }    
\hline
& \ \  \ \ 
&\ \ a \ \ &\ \ b \ \   \\
\hline  
\leftarrow & A
& AB &  \\
\hline  
\rightarrow & B
&  C & A  \\
\hline  
 & \ C \
&   & \ BD \  \\
\hline  
\rightarrow & \ D \
&\  \ & C \\
\hline  
 & \ E \
& D  & \  \  \\
\hline  
\rightarrow & \ F \
&  EF  &  EF \\
\hline  
\end{array}
$
\end{center}
\end{minipage}
\end{table}
\vskip-0.4cm
First, we note that quotients corresponding to the states of $\fD$ can be expressed 
as sets of atoms as follows:
$K_0=\{B,D,F\}$, $K_1=\{C,E,F\}$, $K_2=\{A,C,E,F\}$, $K_3=\{D,E,F\}$,
$K_4=\{B,D,E,F\}$, $K_5=\{A,B,D,E,F\}$, $K_6=\{C,D,E,F\}$, $K_7=\{A,B,C,D,E,F\}$, and
$K_8=\{B,C,D,E,F\}$. One can verify that these are the states of the determinized 
version of the \'atomaton, which is isomorphic to the original DFA $\fD$. 
Now, every state of $\fN_{atom}$ must be a subset of a set of atoms of some quotient, 
and all these sets of atoms of quotients must be covered by the states of $\fN_{atom}$.
We note that quotients $\{B,D,F\}$, $\{C,E,F\}$, and $\{D,E,F\}$
do not contain any other quotients as subsets, while all the other quotients do.
It is easy to see that there is no combination of three or fewer sets of atoms, 
other than these three sets, that can cover these quotients. 
So we have to use these sets as states of $\fN_{atom}$. 
We also need at least one set containing the atom $A$. 
If we use only one set of atoms with $A$, that set has to be
a subset of every quotient having $A$. So it must be
a subset of $\{A,E,F\}$. If we use $\{A\}$ as a state, then by the transition 
table of the \'atomaton, there must be at least one more state to cover 
$\{A,B\}$. Similarly, if we use $\{A,E\}$, then we must have another state to cover 
$\{A,B,D\}$. If we use $\{A,F\}$, then we must have a state to cover 
$\{A,B,E,F\}$. And if we use $\{A,E,F\}$, then we must have a state to cover 
$\{E,F\}$. We conclude that a smallest atomic NFA has at least five states.
There is a five-state atomic NFA, as 
shown in Table~\ref{tab:n5_mp}. It is not unique. 

Since there does not exist a four-state atomic NFA equivalent to the DFA $\fD$,
it is not possible to convert the non-atomic 
minimal NFA $\fN_{min}$ to an atomic NFA by adding transitions.
\qed
\end{proof}

\begin{table}[t]
\vskip-0.5cm
\begin{minipage}[b]{0.3\linewidth}
\caption{NFA $\fN_{min}$.}
\label{tab:n_mp}
\begin{center}
$
\begin{array}{|c|c|| c|c|}    
\hline
 & & \  a \
& \ b \  \\
\hline
\hline
\rightarrow & 0 
& \ 1 \ & 1,2  \\
\hline  
 & 1 
 & 3 & \  0,3 \ \\
\hline  
\leftarrow & \ 2 \ 
& \ 0,2,3 \ &  \\
\hline  
 & \ 3 \ 
& \ 3 \ & 1 \\
\hline  
\end{array}
$
\end{center}
\end{minipage}
\hspace{0.5cm}
\begin{minipage}[b]{0.45\linewidth}
\caption{$\fN_{atom}$.}
\label{tab:n5_mp}
\begin{center}
$
\begin{array}{|c| c||c| c| }    
\hline
& \ \  \ \ 
&\ \ a \ \ &\ \ b \ \   \\
\hline  
\rightarrow & BDF
& CEF &CEF,AEF  \\
\hline  
 & \ CEF \
& DEF  & \ BDF,DEF \  \\
\hline  
\leftarrow & AEF
&  \ BDF, AEF, DEF \ & EF  \\
\hline  
 & \ DEF \
& DEF  & \ CEF \  \\
\hline  \hline
 & \ EF \
&\ DEF  \ & EF \\
\hline  
\end{array}
$
\end{center}
\end{minipage}
\vskip-0.3cm
\end{table}

\vskip-0.1cm
In summary, Sengoku's method cannot find the minimal NFA's in all cases. 
However, it is able to find all atomic minimal NFA's.
His minimization algorithm proceeds by 
``merging some states of the normal nondeterministic automaton.''
This is similar to our search for subsets of atoms that satisfy 
Theorem~\ref{thm:unions}.

\section{The Kameda-Weiner Minimization Method}
\label{sec:KW}

We present a short and modified outline of the properties of the Kameda-Weiner 
NFA minimization method~\cite{KaWe70} using mostly our terminology and notation. 
They consider a trim minimal DFA $\fD=(Q,\Sig,\delta,q_0,F)$ with $Q$ of 
cardinality $n$, and its reversed  determinized and trim version $\fD^{\rev\deter\trim}$; 
the set of states of $\fD^{\rev\deter\trim}$ is a subset $\cS$ of cardinality $p$ of 
$2^Q\setminus\emp$. 
They then 
form an $n\times p$ matrix $T$ where the rows correspond to non-empty states $q_i\in Q$ of $\fD$, 
which is the trim minimal DFA of a language $L$, 
and columns, to states $S_j\in \cS$ of $\fD^{\rev\deter\trim}$, 
which is the trim minimal DFA of the language $L^R$ by Theorem~\ref{thm:Brz}.
The entry $t_{i,j}$ of the matrix $T$ is 1 if $q_i\in S_j$, and 0 otherwise.

We use $\fD^{\rev\deter\rev\trim}$, the trim \'atomaton,  instead of $\fD^{\rev\deter\trim}$, 
since the state sets of these two automata are identical.  
Interpret the rows of the matrix as non-empty quotients of $L$ and columns, 
as positive atoms of $L$. Then $t_{i,j}=1$ if and only if quotient $K_i$ contains 
atom $A_j$, and it is clear that every regular language defines a 
unique such matrix, which we will refer to as the \emph{quotient-atom matrix}.

The ordered pair $(K_i,A_j)$ with $K_i\in \cK$ and $A_i\in \cA$ is a \emph{point} 
of $T$ if $t_{i,j}=1$.
A \emph{grid} $g$ of $T$ is the direct product $g=P\times R$ of a set $P$ of quotients with a set $R$ of atoms.
If $g=P\times R$ and $g'=P'\times R'$ are two grids  of $T$, 
then $g\subseteq g'$ if and only if $P\subseteq P'$ and $R\subseteq R'$.
Thus $\subseteq$ is a partial order on the set of all grids of $T$, 
and a grid is \emph{maximal} if it is not contained in any other grid.
A \emph{cover} $C$ of $T$ is a set $C=\{g_0,\ldots,g_{k-1}\}$ of grids,  
such that every point $(K_i,A_j)$ belongs to some grid $g_i$ in $C$.
A \emph{minimal cover}  has the minimal number
of grids.

Let $f:\cK\to 2^C\setminus \emp$ be the function that assigns to quotient 
$K_i\in \cK$ the set of grids $g=P\times R$ such that $K_i\in P$.
The NFA constructed by the Kameda-Weiner method is $\fN_C=(C,\Sig,\eta_C,C_I,C_F)$, 
where $C$ is a cover consisting of maximal grids, 
$C_I=f(K_{in})$ is the set of grids corresponding to the initial quotient
$K_{in}$, and $C_F$ is defined by $g\in C_F$ if and only if $g\in f(K_i)$ 
implies that $K_i$ is a final quotient.
For every grid $g=P\times R$ and $x\in\Sig$, we can compute 
$\eta_C(g,x)$ by the formula $\eta_C(g,x)=\bigcap_{K_i\in P} f(x^{-1}K_i)$.

It may be the case that $\fN_C$ is not equivalent to DFA $\fD$.
A cover $C$ is called \emph{legal} if $L(\fN_C)=L(\fD)$.
To find a minimal NFA of a language $L$,
the method in~\cite{KaWe70} 
tests the covers of the quotient-atom matrix of $L$ in the order of 
increasing size to see if they are legal. 
The first legal NFA is a minimal one.

When we apply the Kameda-Weiner method~\cite{KaWe70} to the  example in 
Theorem~\ref{thm:Sengoku}, we get the NFA of Table~\ref{tab:n_mp}.

We apply the Kameda-Weiner method~\cite{KaWe70} to the  example in Theorem~\ref{thm:Sengoku}.
The quotients in the example are referred to as the integers 0--8, as in Table~\ref{tab:d_mp}.
The atoms are those in Table~\ref{tab:drdr_mp} relabelled as in Table~\ref{tab:a_mp}. 
The quotient-atom matrix is shown in Table~\ref{tab:c_mp}, where the non-blank entries are to be interpreted as 1's and the blank entries as 0's. 
Table~\ref{tab:c_mp} also shows a minimal cover $S=(g_0,g_1,g_2,g_3)$ and $f(K_i)$ for each quotient $K_i$ of $\cK$.

\begin{table}[hbt]
\caption{Cover $C$ for quotient-atom matrix of $\fD$.}
\label{tab:c_mp}
\begin{center}
$
\begin{array}{|c|c|| c| c| c| c| c| c| c|}    
\hline
& & \  F  \
& \ E \  & \ D \ & \ C \ & \ B \ & \  A  \ & \  f(K_i)  \ \\
\hline
\hline
\rightarrow &\ 0 \ & g_0 &   & g_0 &  & g_0 &  & \{g_0\} \\
\hline  
&\ 1 \ & g_1  & g_1 &  & g_1 &  &  & \{g_1\} \\
\hline  
\leftarrow &\ 2 \ & g_1,g_2  & g_1,g_2 &  & g_1 &  & g_2 & \{g_1,g_2\} \\
\hline  
&\ 3 \ & g_3 & g_3 & g_3 &  &  &  & \{g_3\} \\
\hline  
&\ 4 \ & g_0,g_3  & g_3 & g_0,g_3 &  & g_0 &  & \{g_0,g_3\} \\
\hline  
\leftarrow &\ 5 \ & g_0,g_2,g_3  & g_2,g_3 & g_0,g_3 &  & g_0 & g_2 & \{g_0,g_2,g_3\} \\
\hline  
&\ 6 \ & g_1,g_3  & g_1,g_3 & g_3 & g_1 &  &  & \{g_1,g_3\} \\
\hline  
\leftarrow &\ 7 \ & \ g_0,g_1,g_2,g_3 \ & \ g_1,g_2,g_3 \ &\  g_0,g_3 \ & g_1 & g_0 & g_2 & \
\{g_0,g_1,g_2,g_3\} \ \\
\hline  
& \ 8 \ & g_0,g_1,g_3  & g_1,g_3 & g_0,g_3 & g_1 & g_0 &  & \{g_0,g_1,g_3\} \\
\hline  
\end{array}
$
\end{center}
\end{table}

The construction of the NFA $\fN_{min}$ is shown in Table~\ref{tab:con_mp}.
For each grid $g=P\times R$, we show its set of quotients $P$, with
$K_i\in P$ replaced by $i$.
For each input $x\in\Sig$, we give $x^{-1}P$, and then the intersection 
of the $f(K_i)$ for $K_i\in x^{-1}P$. 
For example, the set $P$ for $g_0$ is expressed as $\{0,4,5,7,8\}$,
the set of quotients $a^{-1}P$ of the set $P$ by $a$ is $\{1,6,7\}$, and 
$\eta_C(g_0,a)=f(1) \cap f(6) \cap f(7) = 
\{g_1\} \cap \{g_1,g_3\} \cap \{g_0,g_1,g_2,g_3\}= \{g_1\}$.
Table~\ref{tab:n_mp} shows the constructed NFA $\fN_{min}$, 
where $g_i$'s are replaced by $i$'s.
Since $\fN_{min}$ is equivalent to $\fD$,  $C$ is a legal cover.
However, $\fN_{min}$ is not atomic, since the right language of 
state $g_2$ is not a union of atoms, although it includes atoms $A$ and $E$ as
its subsets. The right languages of the other states of $\fN_{min}$ are sets of atoms:
$L(g_0)=B\cup D\cup F$, 
$L(g_1)=C\cup E\cup F$, and   
$L(g_3)=D\cup E\cup F$.

\begin{table}
\caption{Construction of NFA $\fN_{min}$.}
\label{tab:con_mp}
\begin{center}
$
\begin{array}{|c| c||c|| c| c|c|c|c|}    
\hline
& \ \ g \ \  & P
&\ \ a \ \ & a &\ \ b \ \  &b  \\
\hline  
 &  & 
 & a^{-1}P &\ \eta_C(g,a) \ & b^{-1}P &\ \eta_C(g,b) \\
\hline
\hline  
\rightarrow & g_0 &\{0,4,5,7,8\}
&  \{1,6,7\} & \{g_1\} & \{2,7\} &\{g_1,g_2\}\\
\hline  
& g_1 &\{1,2,6,7,8\} & \{3,5,6,7\}
&  \{g_3\} &  \{4,7,8\}  & \{g_0,g_3\} \\
\hline  
\leftarrow & g_2 &\{2,5,7\} & \{5,7\}
& \{g_0,g_2,g_3\}& \{2,4,7\} & \emp \\
\hline  
& g_3 &\{3,4,5,6,7,8\} & \{3,6,7\}
&  \{g_3\} &  \{1,2,7,8\}  & \{g_1\} \\
\hline  
\end{array}
$
\end{center}
\end{table}

We believe that NFA's defined by grids are a topic for future research.

\section{Conclusions}
\label{sec:conc}
We have studied the properties of atomic NFA's.
We have shown that atoms play an important role in NFA minimization and
proved that it is not enough to search for atomic NFA's only.


\begin{thebibliography}{10}
\providecommand{\url}[1]{\texttt{#1}}
\providecommand{\urlprefix}{URL }

\bibitem{ADN92}
Arnold, A., Dicky, A., Nivat, M.: A note about minimal non-deterministic
  automata. Bull. EATCS  47,  166--169 (1992)

\bibitem{Brz63}
Brzozowski, J.: Canonical regular expressions and minimal state graphs for
  definite events. In: Proc. Symp. on Mathematical Theory of Automata. MRI
  Symposia Series, vol.~12, pp. 529--561. Polytechnic Institute of Brooklyn,
  N.Y. (1963)

\bibitem{BrTa11}
Brzozowski, J., Tamm, H.: Theory of \'atomata. In: Mauri, G., Leporati, A.
  (eds.) DLT 2011. LNCS, vol. 6795, pp. 105--116. Springer (2011)

\bibitem{BrTa12}
Brzozowski, J., Tamm, H.: Quotient complexities of atoms of regular languages.
  In: Yen, H.C., Ibarra, O. (eds.) DLT 2012. LNCS, vol. 7410, pp. 50--61.
  Springer (2012)

\bibitem{IlYu03}
Ilie, L., Yu, S.: Reducing {NFA}s by invariant equivalences. Theoret. Comput.
  Sci.  306,  373--390 (2003)

\bibitem{KaWe70}
Kameda, T., Weiner, P.: On the state minimization of nondeterministic automata.
  IEEE Trans. Comput.  C-19(7),  617--627 (1970)

\bibitem{MaPo95}
Matz, O., Potthoff, A.: Computing small finite nondeterministic automata. In:
  Engberg, U.H., Larsen, K.G., Skou, A. (eds.) Proc. Workshop on Tools and
  Algorithms for Construction and Analysis of Systems. pp. 74--88. BRICS,
  Aarhus, Denmark (1995)

\bibitem{OtFe61}
Ott, G., Feinstein, N.: Design of sequential machines from their regular
  expressions. J. ACM  8,  585--600 (1961)

\bibitem{Pol05}
Pol\'ak, L.: Minimalizations of {NFA} using the universal automaton. Internat.
  J. Found. Comput. Sci.  16(5),  999--1010 (2005)

\bibitem{RaSc59}
Rabin, M., Scott, D.: Finite automata and their decision problems. IBM J. Res.\
  and Dev.  3,  114--129 (1959)

\bibitem{Sen92}
Sengoku, H.: Minimization of nondeterministic finite automata. Master's thesis,
  Kyoto University, Department of Information Science, Kyoto, Japan (1992)

\end{thebibliography}

%

\end{document}